\documentclass[journal,12pt,onecolumn,draftclsnofoot,]{IEEEtran}

\usepackage{amsmath}
\usepackage{amsfonts}
\usepackage{amssymb}
\usepackage{graphicx}
\usepackage{tikz}
\usepackage{amsthm}

\newtheorem{lemma}{Lemma}

\usepackage{algorithm, algpseudocode, amsthm, amssymb}

\usetikzlibrary{shapes.misc}
\usepackage{caption}
\usepackage{pifont}
\newcommand{\cmark}{\ding{51}}%
\newcommand{\xmark}{\ding{55}}%
\newtheorem{thm}{Theorem}
\theoremstyle{remark}
\usepackage[tight,footnotesize]{subfigure}

\usepackage[tight,footnotesize]{subfigure}

\tikzset{cross/.style={cross out, draw=black, minimum size=2*(#1-\pgflinewidth), inner sep=0pt, outer sep=0pt},
	cross/.default={6pt}}

\usetikzlibrary{matrix,chains,positioning,decorations.pathreplacing,arrows}
\usetikzlibrary{positioning,calc}

\ifCLASSINFOpdf
\else
\fi

\begin{document}
%
\title{Outage Analysis of Energy Efficiency in a Finite-Element-IRS Aided Communication System}

\author{\IEEEauthorblockN{Aaqib Bulla,} and \IEEEauthorblockN{Shahid M Shah}
\IEEEauthorblockA{\\Communication Control \& Learning Lab\\
Department of Electronics \& Communication Engineering\\
National Institute of Technology, Srinagar\\
Email: aaqib.bulla@gmail.com,shahidshah@nitsri.ac.in}}

\maketitle

\begin{abstract}
In this paper, we study the performance of an energy efficient wireless communication system, assisted by a finite-element-intelligent reflecting surface (IRS). With no instantaneous channel state information (CSI) at the transmitter, we characterize the system performance in terms of the outage probability (OP) of energy efficiency (EE). Depending upon the availability of line-of-sight (LOS) paths, we analyze the system for two different channel models, viz. Rician and Rayleigh. For an arbitrary number of IRS elements $(N)$, we derive the approximate closed-form solutions for the OP of EE, using Laguerre series and moment matching methods. The analytical results are validated using the Monte-Carlo simulations. Moreover, we also quantify the rate of convergence of the derived expressions to the central limit theorem (CLT) approximations using the \textit{Berry-Esseen} inequality. Further, we prove that the OP of EE is a strict pseudo-convex function of the transmit power and hence, has a unique global minimum. To obtain the optimal transmit power, we solve the OP of EE as a constrained optimization problem. To the best of our knowledge, the OP of EE as a performance metric, has never been previously studied in IRS-assisted wireless communication systems.
\end{abstract}


%
\IEEEpeerreviewmaketitle

\section{Introduction}\label{sec1}
 
 \begin{table}
	\begin{center}
	\resizebox{\textwidth}{!}{%
		\begin{tabular}{ | c | c | c  | c  | c  | c | c |} 
		
			\hline
			Reference & Channel Model & Rate Outage&Energy Efficiency Outage&No. of IRS elements&Convergence measure&Optimization of OP\\ 
			\hline
			\cite{kudathanthirige2020performance} & Rayleigh  & \cmark&\xmark&Infinite&\xmark&\xmark\\ 
				\hline
			\cite{9146875} & Mixed Rayleigh/Rician  & \cmark&\xmark&Infinite&\xmark&\xmark\\ 
			
			\hline
			\cite{atapattu2020reconfigurable} & Rayleigh & \cmark&\xmark& Finite&\xmark&\xmark\\ 
			\hline
			\cite{boulogeorgos2020performance} & Rayleigh & \cmark&\xmark& Finite&\xmark&\xmark\\ 
			\hline
			\cite{xu2021performance} &Rician/Nakagami-$m$ & \cmark &\xmark& Finite&\xmark&\cmark\\
			\hline
			\cite{samuh2020performance} &Nakagami-$m$ & \cmark &\xmark& Finite&\xmark&\xmark\\
				\hline
			\cite{zhang2019analysis} &Rician & \cmark &\xmark& Finite&\xmark&\cmark\\
				\hline
			\cite{guo2020outage} &Rician & \cmark &\xmark& Finite&\xmark&\cmark\\
				\hline
			\cite{9345753} &Rician& \cmark&\xmark& Finite&\xmark&\xmark\\
			\hline
			\cite{9681955} &Rayleigh & \cmark &\xmark& Finite&\xmark&\xmark\\ 
			\hline
			This work &Rician, Rayleigh& \cmark &\cmark& Finite&Berry-Esseen bound&\cmark\\
			\hline
		\end{tabular}
		}
	\end{center}
\end{table}
 The evolving generations (5G/6G) of wireless communication systems are anticipated to support over 100 billion high-speed communicating devices all over the world by 2030 \cite{strategy2015smarter2030}. Owing to massive user connectivity and ever rising demands for higher data rates, the growth of data traffic has exponentially accelerated over the past decade. It has been estimated that by 2026 the data traffic will rise to 226 exabytes per-month. This drastic progression in the communication technology encompass the challenge of tremendous power consumption. It is reported that the communication technology infrastructure is responsible for 3\% of the worldwide energy consumption and even 2\% of the $CO_2$ emissions \cite{huq2014green}. A major challenge for researchers is to achieve upto 90\% of energy savings in 5G \cite{7918597}. Due to these economical and ecological concerns, the energy efficiency (bits-per-joule) has now emerged as one of the most critical performance measures in the next generation green communication systems \cite{buzzi2016survey}-\cite{wu2017overview}. In particular, several energy efficient solutions have been proposed at the physical layer as well as at the higher layers \cite{hasan2011green}-\cite{tian2020power}. However, the performance of these approaches is still limited by the unregulated medium of propagation.To overcome such demerits the researchers are in a quest of energy efficient solutions for controlling the wireless medium.

 Among the various physical layer technologies, IRS \cite{wu2019towards,gong2020toward} has been acknowledged as a promising technology in the field of energy efficient wireless communication \cite{8741198}. This is primarily due its capability of intelligently reconfiguring the wireless medium with the help of passive reflecting elements \cite{wu2019intelligent,basar2019wireless}. This means that the incident signals can be made to propagate in a desired direction towards the intended receiver. In other words, IRS provides a virtual line-of-sight link in between the transmitter and receiver, resulting in potential signal-to-noise power gain. Specifically, IRS elements can be individually configured to control the phase shifts of incident signals, thereby, beamforming the signals without any energy consumption. Besides a minimal power consumption in the control circuitry, IRS is significantly energy efficient compared to active antenna elements and relays. In fact, the authors in \cite{8741198} proved that, with proper phase adjustments, IRS leads to higher EE as compared to the traditional relays. Furthermore, it was shown in \cite{9322510} that it is possible to achieve higher spectral and energy efficiency in an IRS assisted system without the need of expensive base station antennas. 
 
 For the case of perfect CSI several energy efficient resource allocation and beamforming methods have been studied in the literature. For example, energy efficiency maximization in IRS assisted systems with suboptimal zero-forcing beamforming was studied in \cite{8741198}. A closed form solution of the optimal transmit power for EE maximization was given in \cite{9079457}, wherein, an IRS assisted MISO system with hardware impairments was considered. It was shown in \cite{9257429} that the signal-to-interference plus noise ratio performance in an IRS aided MISO network can be significantly enhanced by jointly optimizing the beamforming vectors at the transmitter and phase shifts at the IRS. However, due to the absence of active radio frequency chain and the fact that the number of passive IRS elements are large, it is very challenging to acquire the CSI perfectly. In literature, some of the works also considered imperfect CSI scenario, for example, power efficient and robust resource allocation design with imperfect CSI was studied in \cite{yu2021irs}.  Several recent works, as listed in table 1, have analyzed the performance of IRS assisted systems by exploiting the statistical CSI. In \cite{kudathanthirige2020performance}, approximate closed form expressions for rate OP, average symbol rate and achievable diversity order has been studied for Rayleigh fading case. It was shown that the given metrics are potentially accurate for large number of IRS elements. For mixed Rayleigh and Rician fading channels, ergodic capacity and rate OP was analyzed in \cite{9146875}. It was found that the IRS contributes a signal-to-noise ratio gain of $N^2$. However, due the central limit theorem (CLT) approach, these methods are applicable only when the number of IRS are sufficiently large. With arbitrary $N$, the performance of IRS aided systems over Rayleigh fading channels was studied in \cite{atapattu2020reconfigurable,boulogeorgos2020performance}. Also, in \cite{9345753}, the authors obtained the approximate closed form expressions for rate OP, average symbol error probability and average channel capacity in case of Rician fading channels. It was shown that the results based on the CLT approximations deviate from the actual behavior, especially for lower values of $N$, emphasizing the need to analyze the system for a finite number of IRS elements. Recently, in \cite{9681955},  Chernoff upper bound on the OP was derived and it was shown that the obtained bound is tighter as compared to CLT approximations. \\Although, a remarkable research work has been presented in the above mentioned papers, none of these have analyzed the OP of EE as a performance measure in the IRS assisted systems. 
 With this motivation and the method of analysis in \cite{BULLA2021}, in this paper, we aim to study the performance of an energy efficient communication system, assisted by a finite-element-IRS in the downlink. However, different from the conventional approaches, we define OP of energy efficiency as the performance metric of the system. Measured in Bits/Hz/Joule, we define EE as ratio of the achievable rate to the total power consumed at transmitter and IRS. The main contribution of this work is summarized as follows:
 \begin{itemize}
 	\item We obtain the approximate closed-form expressions for the OP of EE in both LOS and Non-LOS scenarios for an arbitrary number of IRS elements.
 	\item Next, the rate of convergence of these expressions to the CLT approximations is evaluated using the \textit{Berry-Esseen} inequality. The resulting convergence rate, being a function of $N$, quantifies the deviation of CLT approximations from the derived analytical solutions in terms of the number of IRS elements.
 	\item Further, using the derived expressions, we prove that the OP of EE is a strict pseudo-convex function of transmit power $(p)$. Therefore, in contrast to the conventional rate outage metric, it has a unique global minimum over $p\in (0,P_{max})$ at the critical point. With power budget as the constraint, we formulate the minimization of OP as a constrained optimization problem.
 \end{itemize}
 
 The summary of this paper is as follows: In section 2, we discuss the system model and formulate the problem of OP of EE. In section 3, we obtain the approximate closed form expressions for the OP of EE in LOS and Non-LOS scenarios and also compute the corresponding Berry-Esseen bounds. In section 4, we discuss the optimization of OP of EE. The simulation and analytical results are discussed in section 5. Finally, we conclude the paper in section 6.
\section{System model and Problem formulation}\label{sec2}
We consider an IRS assisted downlink communication system consisting of a single antenna transmitter $T_X$ and a single antenna receiver $R_X$. In this setup, the signal received at the destination can be expressed as follows \cite{basar2019wireless}:
\begin{equation}
r=\Bigg[\sum_{i=1}^{N}h_i e^{j\phi_i}g_i\Bigg]s+n
\end{equation}
where:
\begin{itemize}
	\item $s$ is the transmitted data symbol, with average power $p$.
	\item $h_i=\alpha_i e^{j\nu_i}$ denotes the channel coefficient between the $T_X$ and the $i^{th}$ IRS element having amplitude $\alpha_i$ and phase $\nu_i$.
	\item $g_i=\beta_i e^{j\psi_i}$ denotes the channel coefficient between the $i^{th}$ IRS element and the $R_X$ having amplitude $\beta_i$ and phase $\psi_i$.
	\item $\phi_i$ denotes the phase shift induced by the $i^{th}$ reflecting element.
	\item $n\sim \mathcal{N}(0,N_0)$ models the additive
	white Gaussian noise (AWGN) at the receiver.
\end{itemize} 
Now, depending on how the three nodes ($T_X, R_X$ and IRS) are geographically deployed, we can have different LOS and Non-LOS channel conditions. In particular, we consider the following two cases:
\begin{enumerate}
	\item LOS channels between $T_X$-IRS and IRS-$R_X$.
	\item Non-LOS channels between $T_X$-IRS and IRS-$R_X$.
\end{enumerate}

\subsection{Problem formulation}
As shown in \cite{basar2019wireless}, the maximum end-to-end SNR in the above mentioned IRS-assisted communication link is given by :
\begin{equation}
\zeta=\frac{p}{N_0}\Bigg(\sum_{i=1}^{N}\alpha_i\beta_i\Bigg)^2=\bar{\gamma}\Bigg(\sum_{i=1}^{N}\alpha_i\beta_i\Bigg)^2
\end{equation}
where, $\bar{\gamma}$ denotes the average SNR.\\
Therefore, the energy efficiency of the given system can be defined as follows:
\begin{equation}
\eta_{EE}=\frac{\log_2\left(1+\bar{\gamma}\{\sum_{i=1}^{N}\alpha_i\beta_i\}^2\right)}{p+P_c+P_{IRS}}
\end{equation} 
where, $P_c$ is the static circuit power dissipated at the transmitter during various signal processing operations and $P_{IRS}$ is the hardware static power dissipated at the IRS.

Now, in order to characterize the system performance, we define the OP of EE as a performance measure. Corresponding to a given target value of EE $(\eta_{th})$, the OP of EE is given as follows:
\begin{equation}
P_{out}(\eta_{th})=P(\eta_{EE} < \eta_{th})
= P\left[\frac{\log_2\left(1+\zeta\right)}{P_{T}} < \eta_{th}\right]
\end{equation}
where, $P_T=p+P_c+P_{IRS}$.\\
A physical interpretation of eq (4) is that, over a unit bandwidth, the channel is allowing a data rate of $\log_2(1+\zeta)$ bits/s per unit of energy consumed at the transmitter. Thus, a reliable and energy efficient communication is possible if this rate per joule exceeds the target EE. Otherwise, the system is said to be in outage. The optimum power level that minimizes this outage represents the best trade-off in-between the achievable rate and saving the energy consumption. 

For the given setup, the rate OP for a given target rate $R_{th}$ can be written as:
\begin{align}
P_{out}(R_{th})&=P\left[\log_2\left(1+\zeta\right) < R_{th}\right]\nonumber\\
&=P[\zeta<2^{R_{th}}-1]
\end{align}
Also. eq (4) can be re-written as:
\begin{equation}
P_{out}(\eta_{th})=P[\zeta<2^{\eta_{th}P_T}-1]
\end{equation}
Corresponding to eq (5) and eq (6), there are two power profiles as given below:
\begin{align}
P^R&=\{p:\zeta<2^{R_{th}}-1\}\\
P^{EE}&=\{p:\zeta<2^{\eta_{th}P_T}-1\}
\end{align}

From eq. (7) and (8) it is clear that, in contrast to the rate outage metric, the power profile corresponding to the OP of EE depends on the product of the threshold value and the total power consumption. Hence, different approaches need to be followed for the analysis and optimization of the two metrics. Contrary to the rate OP, which is typically minimized by using the maximum available transmit power, the OP of EE exhibits a unimodal behavior. This is illustrated by plotting both the rate OP and OP of EE as a function of transmit power in Figure 1.
\begin{figure}[h]
	\centering
	\includegraphics[width=.858\columnwidth]{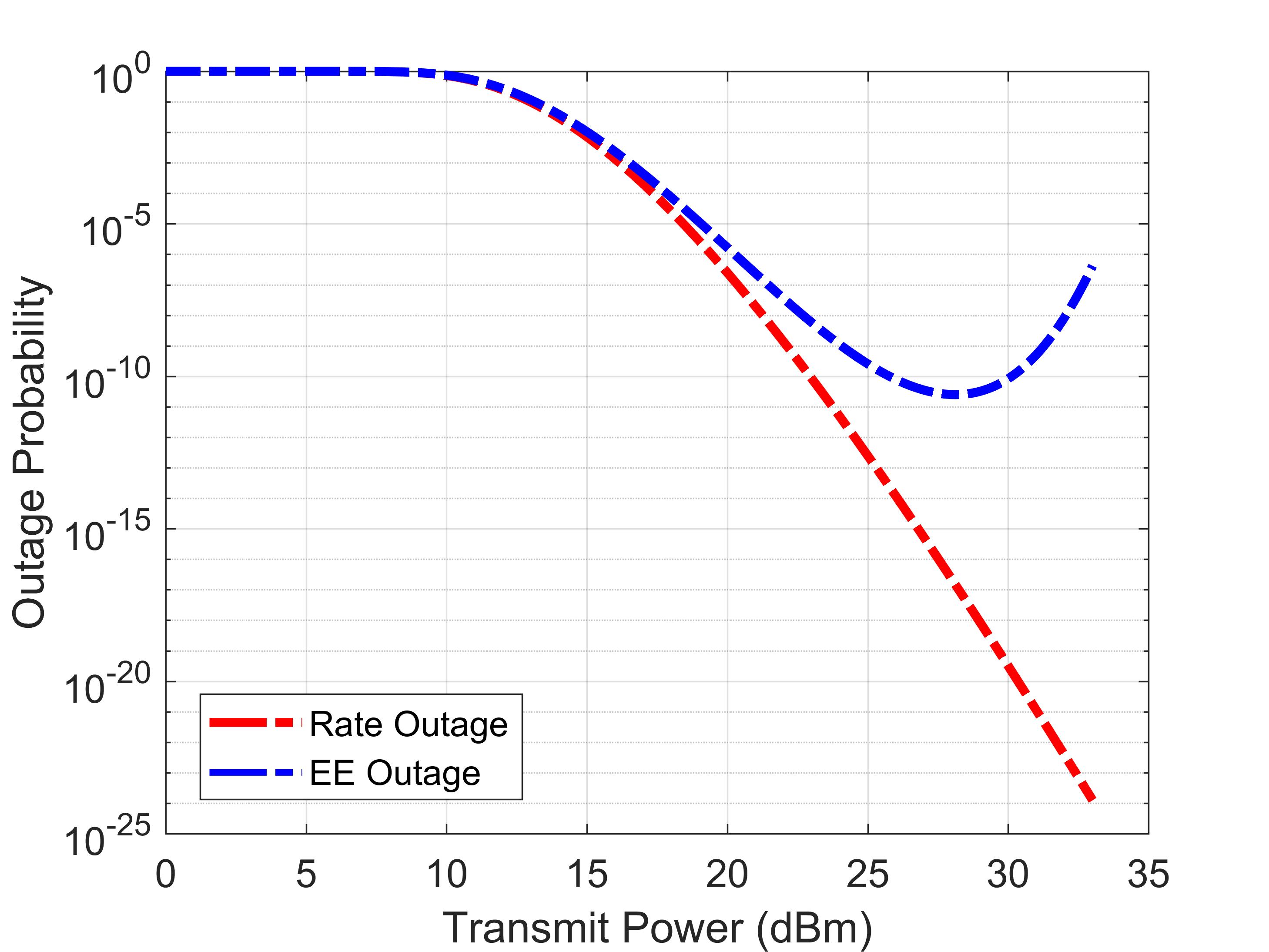}
	\caption{Rate outage and EE outage as a function of transmit power}
\end{figure}
 As evident, the power level that minimizes the OP of EE is essentially lower than the maximum available transmit power. \\
Now, eq (4) can be written as: 
\begin{equation} \nonumber
\begin{split}
P_{out}(\eta_{th})= P&\left(\bigg\{\sum_{i=1}^{N}\alpha_i\beta_i\bigg\}^2 < \frac{\{2^{(P_T)\eta_{th}}\}-1}{\bar{\gamma}}\triangleq Q\right)\\
=P&\left(\bigg\{\sum_{i=1}^{N}\alpha_i\beta_i\bigg\} < \sqrt{Q}\right)
=F_X(\sqrt{Q})
\end{split}
\end{equation}
where, $X=\sum_{i=1}^{N}\alpha_i\beta_i=\sum_{i=1}^{N}x_i$\\
This implies that we need to compute the cumulative distribution function (CDF) of $X$ in order to quantify the OP of energy efficiency for the given system. Clearly, the CDF of $X$ will vary according to the distribution of channel coefficients $\alpha_i$, $\beta_i$ for the above mentioned cases.
\section{Outage Probability of Energy Efficiency}
\subsection{Case 1: LOS channels between $T_X$-IRS and IRS-$R_X$.}
The IRS is preferably deployed in such a way that an LOS path is established between the $T_X$ and IRS as well as between the IRS and $R_X$. In such a case, the envelops ($\alpha_i$, $\beta_i$) of channel coefficients are conventionally modeled using independent and identically distributed (IID) Rician random variables, with probability density function (PDF) as given below:
\begin{equation}
f(u)=\frac{u}{\hat{\sigma}}\exp{\Big(-\frac{u^2+A^2}{2\hat{\sigma}^2}\Big)}I_0\Big(\frac{uA}{\hat{\sigma}^2}\Big) ; u>0
\end{equation}
 where, $\hat{\sigma}^2$ denotes power of NLOS components and $A^2$ denotes power of LOS component. Let $K_1$, $K_2$ and $\Omega_1$, $\Omega_2$ represent the corresponding shape and scale parameters of the Rician variables $\alpha_i$ and $\beta_i$ respectively. Assuming that $\alpha_i$ and $\beta_i$ are independent, the PDF of $x_i=\alpha_i\beta_i$ is therefore given by \cite{bhargav2018product}:
\begin{equation}
f_{x_i}(y)=\sum_{j=0}^{\infty} \sum_{i=0}^{\infty} \frac{\left(4 K_{2}^{i} K_{1}^{j}\left(\Omega_{1} \Omega_{2}\right)^{\frac{(i+j+2)}{2}}\right) K_{j-i}\left(2 y \sqrt{\Omega_{1} \Omega_{2}}\right)}{(y^{-(i+j+1)})(i !)^{2}(j !)^{2} \exp \left(K_{1}+K_{2}\right)}
\end{equation}
where, $K_v(.)$ is the modified $v$-order Bessel function of the
second kind. Also, the mean and variance of $x_i$ are given by \cite{bhargav2018product}:
\begin{align}\nonumber
E\left[x_i\right]=& \frac{\pi e^{\frac{\left(-K_{1}-K_{2}\right)}{2}}\sqrt{\Omega_{1}\Omega_{2}}}{4 \sqrt{(K_1+1)(K_2+1)}}J_1J_2 \nonumber\\
Var\left[x_i\right]=& \frac{\left[16\left(K_{1}+1\right)\left(K_{2}+1\right)-\pi^{2} e^{-K_{1}-K_{2}}(J_1^2J_2^2)\right]}{16 \Omega_{1} \Omega_{2}}\nonumber
\end{align}
where, $J_1=\left[\left(K_{1}+1\right) I_{0}\left(\frac{K_{1}}{2}\right)+K_{1} I_{1}\left(\frac{K_{1}}{2}\right)\right]$, $J_2=\left[\left(K_{2}+1\right) I_{0}\left(\frac{K_{2}}{2}\right)+K_{2} I_{1}\left(\frac{K_{2}}{2}\right)\right]$ and $I_v(.)$ is the modified $v$-order Bessel function of the first kind.
\begin{thm}
	The OP of EE, when $\alpha_i$ and $\beta_i$ are assumed to be rician distributed, is given by:
	\begin{equation}
	P\big(X<\sqrt{Q}\big)=F_{X}\big(\sqrt{Q}\big)=\frac{1}{\Gamma(a)}\gamma\bigg(a,\frac{1}{b}\sqrt{Q}\bigg)
	\end{equation}
	where, $\gamma(.,.)$ is the lower incomplete gamma function and \\
	$a=\frac{N[E(x_i)]^2}{Var(x_i)}$, $b= \frac{Var(x_i)}{E(x_i)}$
\end{thm}

\begin{proof}
	We have $X=\sum_{i=1}^{N}\alpha_i\beta_i=\sum_{i=1}^{N}x_i$. Since $X$ is the sum of $N$ IID random variables, according to \cite{primak2005stochastic}, its PDF can be tightly approximated by the first term of Laguerre series expansion as following:
	\begin{equation}
	f_{X}(x)=\frac{x^{a-1}}{b^{a} \Gamma(a)} \exp \left(-\frac{x}{b}\right)
	\end{equation}
	where, $a$ and $b$ are related to the mean and variance of $X$ as $a=\frac{[E(X)]^2}{Var(X)}$ and $b= \frac{Var(X)}{E(X)}$. Now, using the fact that $x_i$ are independent and that $E[.]$ is a linear operator, we can write $E(X)=NE(x_i)$ and $Var(X)=NVar(x_i)$. Now, from (12), $P_{out}(\eta_{EE})$ can be easily computed as:
	\begin{align}
	P_{out}(\eta_{EE})&=F_X\big(\sqrt{Q}\big)=\hspace{-0.3em}\int_{0}^{\sqrt{Q}}\hspace{-0.5em}f_X(x)dx=\frac{\gamma(a,\frac{1}{b}\sqrt{Q})}{\Gamma(a)}
	\end{align}
\end{proof}
\subsubsection{Large $N$}
When the number of IRS elements are sufficiently large, i.e, $N \rightarrow \infty$, then by CLT, $X$ can be approximated by a Gaussian random variable with mean $\mu=N[E(x_i)]$ and variance $\sigma^2=NVar(x_i)$. Therefore, the CDF of $X$ can be wrintten as:
\begin{equation}
\Psi_{X}(x)=\Phi_X\bigg(\frac{x-\mu}{\sigma}\bigg)=\frac{1}{2}\left(1+\operatorname{erf}\left[\frac{x-\mu}{\sqrt{2 \sigma^2}}\right]\right)
\end{equation}
where, $\Phi_X(.)$ denotes the CDF of standard gaussian random variable. The $P_{out}(\eta_{EE})$ is therefore given by:
\begin{equation}
P(X<\sqrt{Q})=\Psi_X(\sqrt{Q})=\Phi_X\bigg(\frac{\sqrt{Q}-\mu}{\sigma}\bigg)
\end{equation}

\subsection{Case 2: Non-LOS channels between $T_X$-IRS and IRS-$R_X$.}
Now, let's consider a scenario when the direct LOS paths are unavailable due to various blockage effects. In such a case, we assume that all channels are IID complex Gaussian fading with zero mean and unit variance, i.e, $h_i$,$g_i$ $\sim$ $\mathcal{CN}(0,1)$  
It follows that $\alpha_i$ and $\beta_i$ are Rayleigh distributed with PDF as given below:
\begin{equation}
f(v)=\frac{v}{\tilde{\sigma}}\exp(-\frac{v^2}{2\tilde{\sigma}}) ; v>0
\end{equation}
where, $\tilde{\sigma}$ denotes the power of multi-path components.  Let $\hat{\rho_{\alpha}}=\hat{\rho_{\beta}}=\frac{1}{\sqrt{2}}$ denote the corresponding scale parameters of the Rayleigh variables $\alpha_i$ and $\beta_i$. Again, to evaluate the OP of EE, we need to compute the CDF of $X$ for the given case.
\begin{thm}
	The OP of EE, when $\alpha_i$ and $\beta_i$ are assumed to be Rayleigh distributed, is given by:
	\begin{equation}
	P(X<\sqrt{Q})=F_X\big(\sqrt{Q}\big)=\frac{1}{\Gamma(Nk)}\gamma\bigg(Nk,\frac{1}{\theta}\sqrt{Q}\bigg)
	\end{equation}
	where; $k=\frac{\pi^2}{16-\pi^2}$ and $\theta=\frac{16-\pi^2}{4\pi}$
\end{thm}
\begin{proof}
	The PDF of the product of two i.i.d.
	Rayleigh random variables $(x_i=\alpha_i\beta_i)$ is given by \cite{salo2006distribution}:
	\begin{equation}
	f_{x_i}(y)=\frac{y}{\hat{\rho}^2}K_0\bigg(\frac{y}{\hat{\rho}}\bigg)
	\end{equation}
	where, $K_0$ is the zeroth order modified Bessel function of second kind and $\hat{\rho}=\hat{\rho_{\alpha}}\hat{\rho_{\beta}}$. Also, the mean and variance of $x_i$ are given by; $E(x_i)=\frac{\pi}{4}$ and $Var(x_i)=1-\frac{\pi^2}{16}$. Now, using the moment matching technique \cite{6059452}, we can use the Gamma distribution to approximate the distribution of $x_i$. A Gamma random variable with shape parameter $k$ and scale parameter $\theta$ has mean $k\theta$ and variance $k\theta^2$. Therefore, on matching the given moments as shown below:
	\begin{equation}
	\begin{aligned}
k\theta=\frac{\pi}{4}\\
k\theta^2=1-\frac{\pi^2}{16}
	\end{aligned}
	\end{equation}
	the CDF of $x_i$ can be written as:
	\begin{equation}
	F_{x_i}(y)=\frac{\gamma\big(k,\frac{1}{\theta}y\big)}{\Gamma(k)}
	\end{equation} 
	where, $k=\frac{\pi^2}{16-\pi^2}$ and $\theta=\frac{16-\pi^2}{4\pi}$. Now, $X$ is the sum of $N$ IID Gamma random variables with parameters $k$ and $\theta$. Therefore random variable $X$ is also Gamma distributed with parameters $Nk$ and $\theta$. Hence $P_{out}$ is given by (17).
\end{proof}
\subsubsection{Large $N$}
Again, if the number of IRS elements are sufficiently large, i.e, $N \rightarrow \infty$, according to CLT, $X$ converges to a Gaussian distributed random variable. Therefore, following the similar notation as in case 1, the $P_{out}(\eta_{EE})$ is given by:
\begin{equation}
P\big(X<\sqrt{Q}\big)=\Psi_X(\sqrt{Q})=\Phi_X\bigg(\frac{\sqrt{Q}-\mu}{\sigma}\bigg)
\end{equation}
In this case, $\mu=N\big[\frac{\pi}{4}\big]$ and $\sigma^2=N[1-\frac{\pi^2}{16}]$. 
\subsection{Berry Esseen bound}
As mentioned, the approximations in (15) and (21) are valid for sufficiently large number of reflecting elements. Now, in order to precisely quantify the effect of $N$ over the rate of convergence of CDF of $X$ to the Gaussian CDF we invoke the Berry-Esseen inequality \cite{durrett2019probability} as follows:\\
\begin{equation}
\begin{split}
&\bigg|F_{X}\bigg(\frac{\sqrt{Q}-\mu}{\sigma}\bigg)-\Phi_X\bigg(\frac{\sqrt{Q}-\mu}{\sigma}\bigg)\bigg|\leq \frac{c\rho}{\{Var(x_i)\}^{3/2}\sqrt{N}}\\
\end{split}
\end{equation}
where, $\rho=E[\{x_i-E[x_i]\}^3]=E[x_i^3]-3E[x_i]E[x_i^2]+3[E[x_i]]^3$ and the best known value of $c$ is 0.56 \cite{shevtsova2010improvement}. \\
\begin{itemize}
	\item For case 1 $E[x_i^k]$ is given by \cite{bhargav2018product}:
	\begin{align}
	E[x_i^k]&=\frac{\hat{r}_1^k\hat{r}_2^k[\Gamma(1+\frac{k}{2})]^2e^{-K_1-K_2}}{[(1+K_1)(1+K_2)]^{\frac{k}{2}}}\nonumber\\ &\times{}_{1}F_{1}\bigg(1+\frac{k}{2};1;K_1\bigg){}_{1}F_{1}\bigg(1+\frac{k}{2};1;K_2\bigg)
	\end{align}
	where, ${}_{1}F_{1}(.;.;.)$ denotes the confluent hypergeometric function and $\hat{r}_1=\sqrt{E[\alpha_i^2]}$, $\hat{r}_2=\sqrt{E[\beta_i^2]}$.
	\item For case 2 $E[x_i^k]$ is given by \cite{salo2006distribution}: 
	\begin{equation}
	E[x_i^k]=(4\rho^4)^\frac{k}{2}\bigg[\Gamma\bigg(\frac{k}{2}+1\bigg)\bigg]^2
	\end{equation}
\end{itemize}
As evident from eq. (22), the approximation error is directly proportional to $\frac{1}{\sqrt{N}}$ and is upper bounded by $\mathcal{O}(\frac{1}{\sqrt{N}})$, $N$ being the number of IRS elements.

\section{Optimization of OP}
The closed form expressions for the OP of EE for case 1 and 2 are given by eq (11) and eq (17) respectively. When plotted with respect to (w.r.t) the transmitted power, we observe that OP of EE does not asymptotically go to zero, rather it increases for higher values of transmit power. This is primarily due to the dominance of the linear function in the denominator of (3) beyond a certain point. In this section, we aim to optimize the OP of EE w.r.t to the transmit power. The optimization problem can be formulated as given below:
\begin{align}
\min_{p} \frac{\gamma\big(\varphi,\frac{1}{\lambda}\sqrt{Q}\big)}{\Gamma(\varphi)}
\hspace{3em}\text{s.t:} \hspace{1em}0<p\leq P_{max}
\end{align}
where, for case 1: $\varphi=a+1$, $\lambda=b$ and for case 2: $\varphi=NK$, $\lambda=\theta$. $P_{max}$ denotes the maximum transmit power.\\
The objective function is (25) is of the form as given below:
\begin{align}
f(\omega)=g(h(\omega))
\end{align}
where, $h(\omega)=\sqrt{\frac{2^{k(\omega+c)}-1}{\omega}}$ and $g(\omega)=\frac{\gamma(\varphi,\frac{1}{\lambda}\omega)}{\Gamma(\varphi)}$

\begin{lemma}
	$h(\omega)$ in (26) is a strictly convex function of $\omega$ for $\omega>0$.
\end{lemma} 
\begin{proof}
	let $z(\omega)=\log\{h(\omega)\}$\\
	\begin{equation}\nonumber
	\begin{split}
	\implies z(\omega)=\frac{1}{2}\Big[k(\omega+c)\log(2)-\frac{1}{2}\log(\omega)\Big]\\
	\end{split}
	\end{equation}
	Differentiating twice w.r.t $\omega$ we have: $z''(\omega)=\frac{1}{2\omega^2}$.
	Clearly $z''(\omega)>0$ for $\omega>0$, therefore, $h(\omega)$ is a strict log-convex function of $\omega$. Hence, $h(\omega)$, being a composite function of an exponential function and a strict log-convex function is strictly convex in $\omega$. 
\end{proof}
\begin{lemma}
	The function $g(\omega)$ is a quasi-convex function of $\omega$, for $\omega>0$. 
\end{lemma} 
\begin{proof}
	Since $g(\omega)$ is the CDF of gamma distribution, it is a monotonically increasing function of $\omega$. Therefore, $\forall t\in[0,1]$ the relation:
	\begin{equation}
	g(t\omega_1+(1-t)\omega_2)\leq \max [g(\omega_1), g(\omega_2)]
	\end{equation}
	is true. Hence, according to (\cite{stancu2012fractional}, def. 2.2.2) it is a quasi-convex function of $\omega$.
\end{proof}
\begin{thm}
	The objective function in (25) has a unique global minimum over $p\in(0, P_{max})$.
\end{thm}  
\begin{proof}
	
	From lemma 1 and lemma 2 we note that $f(\omega)$ is composite function with $h(\omega)$ strictly convex and $g(\omega)$ quasi-convex. Thus according to (\cite{stancu2012fractional}, theorem 2.4.1) it is a strict pseudo-convex function of $\omega$. Hence, by virtue of pseudo-convexity, the objective function in (25) has a unique global minimum at its critical point.
\end{proof}
Since, we have proved that the objective function is strict pseudo-convex, the optimization problem can be solved by convex optimization techniques. In particular, since the function is twice differentiable, we solve the KKT conditions \cite{boyd2004convex} by the method of sequential quadratic programming \cite{nocedal2006numerical}.
\section{Results and discussion}
In this section, we present the numerical results to verify the theoretical framework discussed in section III and IV.\\ Figure 2 illustrates the OP of EE as a function of transmit power (dBm) for both LOS and Non-LOS case, with $N=4$. We observe that the OP of EE decreases with the transmit power until it achieves a global minimum, and then, tends to increases again. Therefore, we infer that a continuous increment in the transmit power does not guarantee an improved system performance, unlike rate OP, as shown in Figure 3.
\begin{figure}[h]
	\centering
	\includegraphics[width=.78\columnwidth]{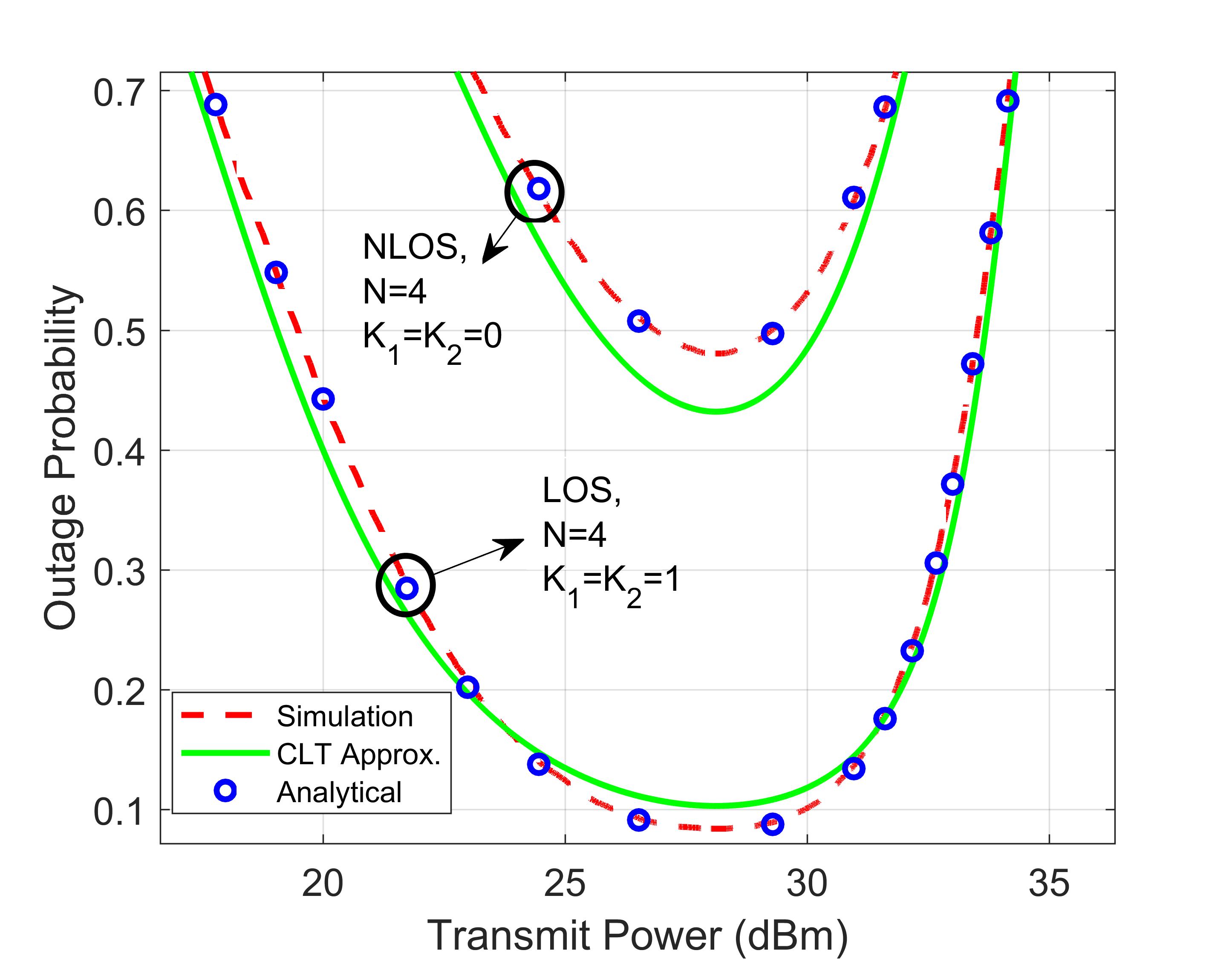}
	\caption{$P_{out}(\eta_{EE})$ vs $p$ for LOS and Non-LOS} 
\end{figure}
\begin{figure}[h]
	\centering
	\includegraphics[width=.78\columnwidth]{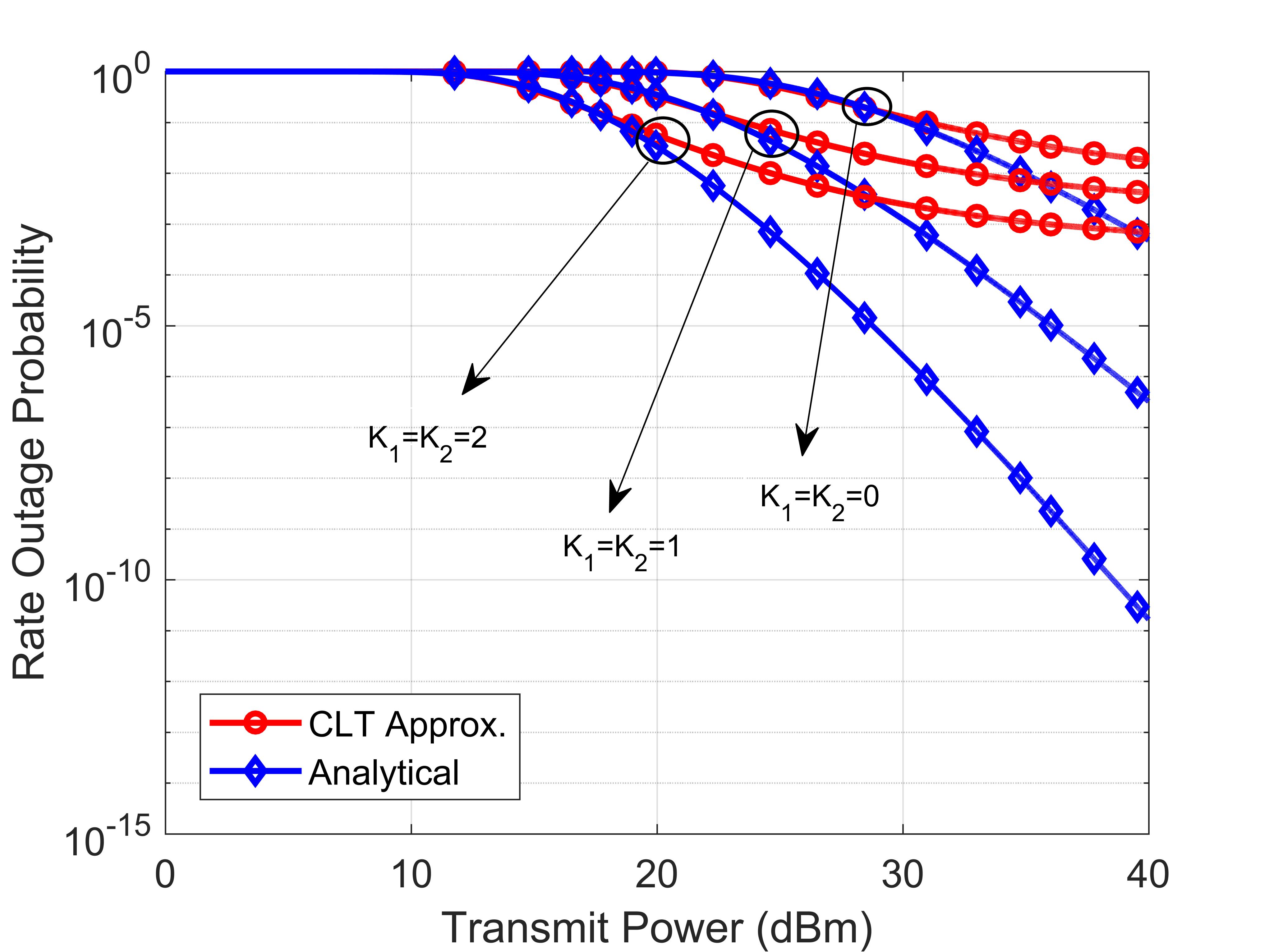}
	\caption{Rate OP as a function of transmit power $p$} 
\end{figure}
 In Figure 2, the analytical results and CLT approximations are compared with the Monte-Carlo simulations. Clearly, the results corresponding to the analytical expressions match with the simulation results. To this end, it is also evident in Figure 2 that the CLT approximations are not very accurate for smaller values of $N$. However, as given by the \textit{Berry-Esseen} inequality, the analytical expressions converge to CLT approximations for higher values of $N$, with the rate of convergence being proportional to $\frac{1}{\sqrt{N}}$. This is illustrated in Figure 4, wherein, the approximation error is plotted as a function of $N$, for $\eta_{th}=2$ Bits/Hz/J. Further, as expected, when the LOS paths exist in-between the three nodes, the system performs significantly better. The impact of LOS components is quantified in terms of the Rician factors $(K_1, K_2)$, as shown in Figure 5. The Rician parameter $K_i$, $i\in {1,2}$, is the ratio of power in the LOS component and Non-LOS components. Therefore, a higher value of $K_i$ implies the reduced channel randomnesses and hence leads to an improvement in the system performance. Moreover, the OP significantly relies on the number of IRS elements $(N)$. In Figure 6, the OP of EE is plotted w.r.t $N$ for increasing values of $K_i$. We observe that the value sharply decreases with the increasing number of IRS elements and the effect is more significant in the case of dominant LOS-component scenarios. In fact, increasing the value of $N$ reduces the transmit power required to achieve a given OP.  From the results we also infer that a fairly small value of OP is achieved even with a finite number of IRS elements and the assumption of infinitely large number of IRS elements is not necessary. 

\begin{figure}[h]
	\centering
	\includegraphics[width=.78\columnwidth]{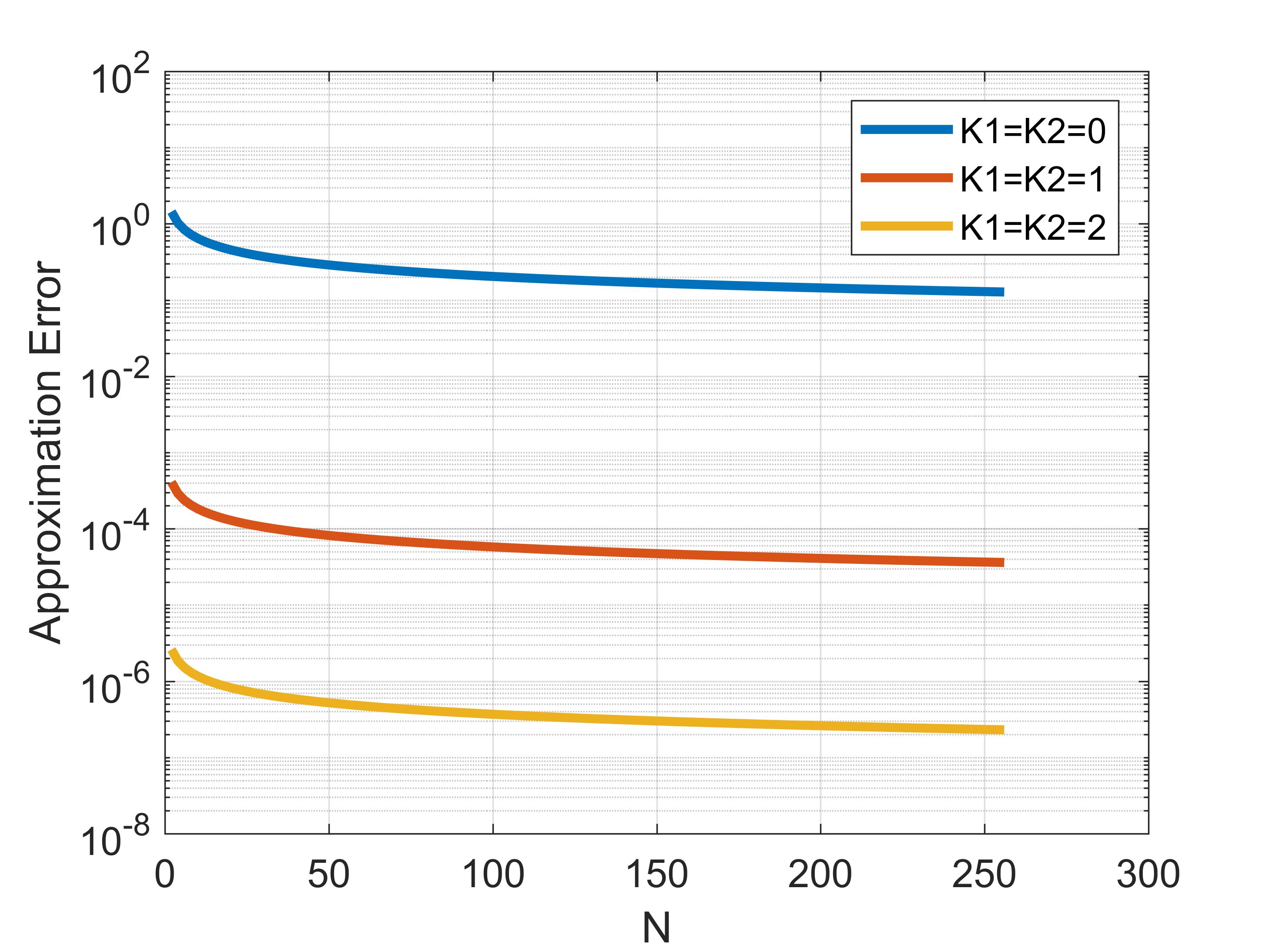}
	\caption{Approximation Error vs $N$ for $\eta_{th}=2$ bits/Hz/J}
\end{figure}
\begin{figure}[h]
	\centering
	\includegraphics[width=.78\columnwidth]{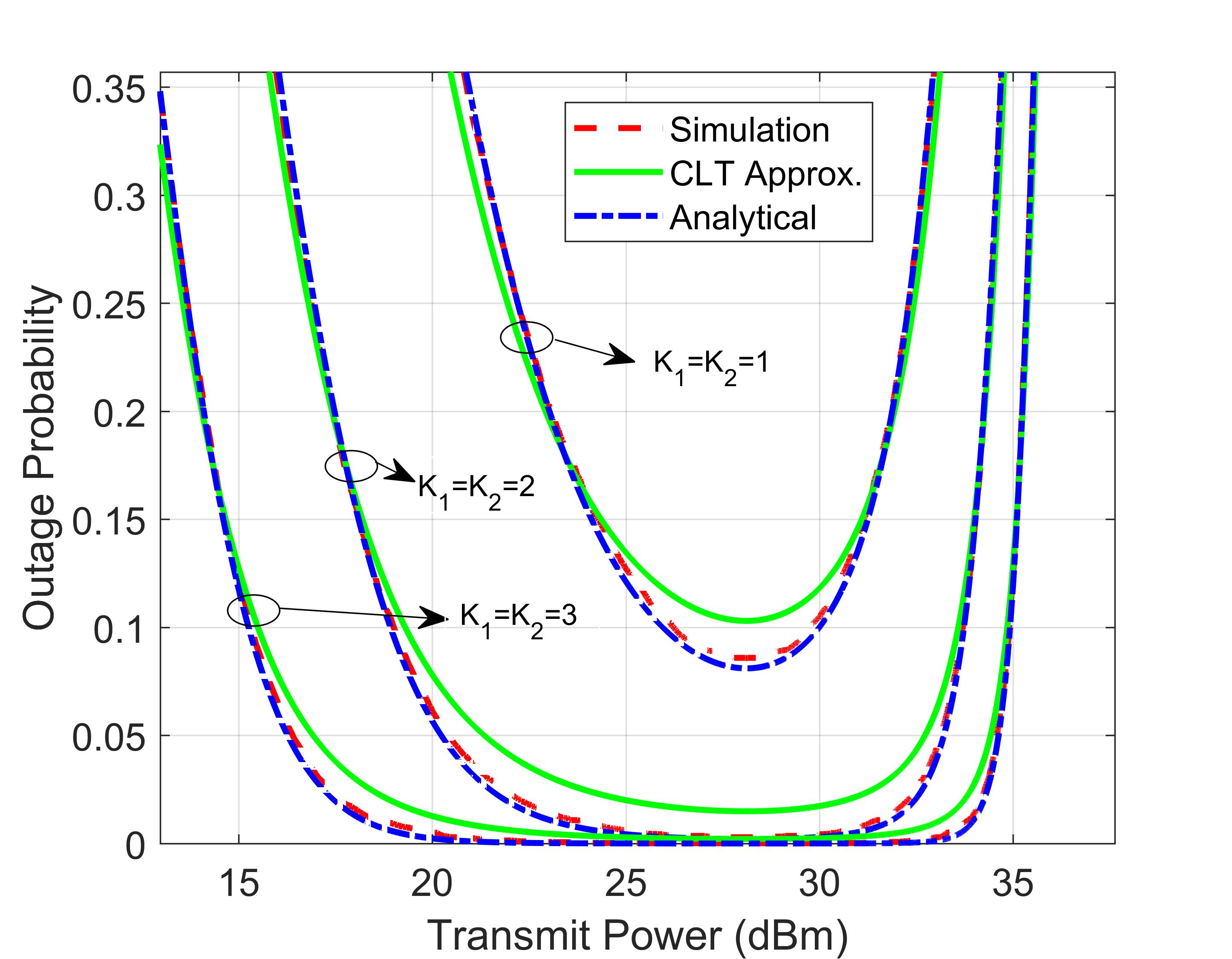}
	\caption{$P_{out}(\eta_{EE})$ vs $p$ for different values of $K_1$, $K_2$}
\end{figure}
\begin{figure}[h]
	\centering
	\includegraphics[width=.78\columnwidth]{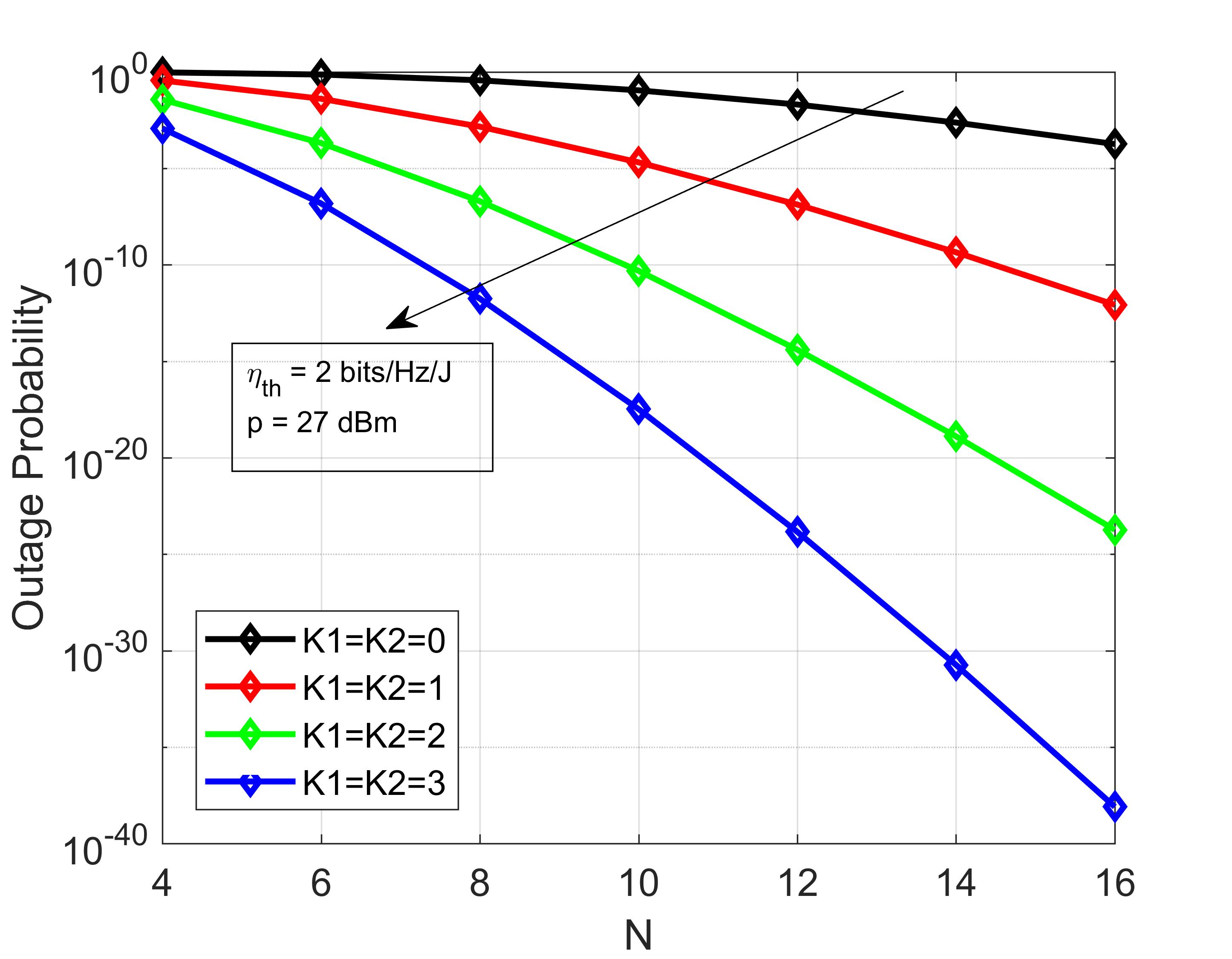}
	\caption{$P_{out}(\eta_{EE})$ vs $N$ for $p=27$ dBm and $\eta_{th}=2$ bits/Hz/J}
\end{figure}
\begin{table}[h!]
	\centering
	\caption{Required IRS elements in LOS and non-LOS scenario}
	\begin{tabular}{||c c c c||} 
		\hline
		Optimal Value of OP & $K_i=1$ & $K_i=0$ & $(p)$ \\ [0.5ex] 
		\hline\hline
		9.4$\times$$10^{-06}$ & $N=8$ & $N=14$ & 28 dBm \\
		3.6$\times$$10^{-11}$ & $N=12$ & $N=20$ & 28 dBm \\ 
		9.1$\times$$10^{-18}$ & $N=16$ & $N=26$ & 28 dBm \\
		2.7$\times$$10^{-25}$ & $N=20$ & $N=32$ & 28 dBm \\ [1ex]
		\hline
	\end{tabular}
\end{table}
Lastly, as indicated in Table 1, the number of IRS elements required to achieve a given minimum OP, in LOS and non-LOS scenarios, are approximately in the ratio of 5:8. Thus, in absence of LOS paths, the IRS, with nearly twice the number of elements, can compensate for the fading effects in an energy efficient manner.
\section{Conclusion}
In this paper, the EE of a finite-element-IRS assisted communication system was studied. With no CSI at the transmitter, we defined the OP of EE as a metric to characterize the system performance. For the LOS and Non-LOS scenarios, we obtained the closed-form solutions for the OP of EE and also quantified the rate of convergence of the derived expressions using the \textit{Berry-Esseen} inequaliy. It was shown that the rate of convergence is proportional to $\frac{1}{\sqrt{N}}$. Further, we proved that the OP of EE is a strict-pseudoconvex function of transmit power and has unique global minimum over $p\in (0,P_{max})$. To obtain the optimal power, we formulated  a constrained optimization problem and solved the KKT conditions. The numerical results illustrated the impact of $K$ and $N$ on the OP of EE. It was observed that the IRS can  efficiently compensate for the fading effects in non-LOS channels. The analytical results were verified using the Monte-Carlo simulations.

\bibliographystyle{IEEEtran}
\bibliography{Outage_IRS_EE}

\end{document}